\documentclass{article}
\usepackage{ijcai18}

\usepackage{times}
\usepackage{xcolor}
\usepackage{soul}
\usepackage[utf8]{inputenc}
\usepackage[small]{caption}

\usepackage{amsmath}
\usepackage{amsfonts}
\usepackage{amssymb}
\usepackage{amsthm}
\usepackage{graphicx}
\usepackage{hyperref}
\usepackage{bm}
\usepackage{mathtools}
\usepackage{enumitem}
\usepackage{leftidx}
\usepackage{isomath}
\usepackage{stmaryrd}

\usepackage[small]{complexity}
\usepackage{xspace}
\usepackage{algorithm}
\usepackage{algpseudocode}
\usepackage{varwidth}

\usepackage{colortbl}

\usepackage{balance}
\usepackage{environ}

\def\e#1{\emph{#1}}
\def\val#1{\texttt{#1}}

\def\attname#1{\textsf{#1}}
\def\relname#1{\text{\sc #1}}
\def\dla{\mathrel{\,{:}{-}}\,}
\def\Dc{\widehat D}

\newcommand{\PW}{\mbox{PW}}
\newcommand{\NW}{\mbox{NW}}
\newcommand{\NA}{\mbox{NA}}
\newcommand{\PA}{\mbox{PA}}
\newcommand{\Win}{\relname{Winner}}

\def\tup#1{\mathbf{#1}}

\def\w#1{{\cellcolor[gray]{0.5}\color{white}\textsf{#1}}}

\newenvironment{proofhint}
{\par\noindent\textbf{Hint of Proof:}\,}
{\par}

\def\angs#1{\langle#1\rangle}

\def\necess#1{\mathsf{Necessity}(#1)}
\def\possib#1{\mathsf{Possibility}(#1)}

\newtheorem{definition}{Definition}

\newtheorem{lemma}{Lemma}
\newtheorem{theorem}{Theorem}
\newtheorem{corollary}{Corollary}

\newtheorem{example}{Example}

\newcommand*{\eat}[1]{}

\usepackage{color}

\RequirePackage{mathpazo}

\RequirePackage{fancyhdr}
\NewEnviron{omitframe}{}

\def\set#1{\{#1\}}
\def\score{\mathrm{score}}

\title{Computational Social Choice Meets Databases\thanks{This is an
    extended version of ``Computational Social Choice Meets
    Databases'' by Kimelfeld, Kolaitis and Stoyanovich, to appear in
    IJCAI 2018.}}

\author{
Benny Kimelfeld$^1$, 
Phokion G. Kolaitis$^2$, 
Julia Stoyanovich$^3$
\\ 
$^1$ Technion, Israel \\
$^2$ UC Santa Cruz and IBM Research-Almaden, USA\\
$^3$ Drexel University, USA  \\
bennyk@cs.technion.ac.il,
kolaitis@cs.ucsc.edu,
stoyanovich@drexel.edu 
}

\begin{document}

\maketitle
%
\begin{abstract}
  We develop a novel framework that aims to create bridges between the
  computational social choice and the database management
  communities. This framework enriches the tasks currently supported
  in computational social choice with relational database context,
  thus making it possible to formulate sophisticated queries about
  voting rules, candidates, voters, issues, and positions. At the
  conceptual level, we give rigorous semantics to queries in this
  framework by introducing the notions of necessary answers and
  possible answers to queries. At the technical level, we embark on an
  investigation of the computational complexity of the necessary
  answers. We establish a number of results about the complexity of
  the necessary answers of conjunctive queries involving positional
  scoring rules that contrast sharply with earlier results about the
  complexity of the necessary winners.
\end{abstract}

\section{Introduction}
\label{sec:intro}

Social choice theory is concerned with the aggregation of preferences
expressed by the members of a society to arrive at a collective
decision.  The origins of social choice theory are often traced to the
work of Jean-Charles de Borda and Marquis de Condorcet in the 18th
Century, even though it is now known that Condorcet's voting rule had
been already proposed by Ramon Llull in the 13th Century
\cite{hagele2001llull}.  During the past two decades, social choice
theory has been examined under the algorithmic lens, and computational
social choice (COMSOC) has emerged as an interdisciplinary research
area that combines insights and methods from mathematics, economics,
logic, and computer science.  The COMSOC community has carried out an
in-depth investigation of computational aspects of voting and
preference aggregation in an election or a poll.  Since preferences
are often only partially expressed, the notions of \emph{necessary
  winners} and \emph{possible winners} were formulated by Konczak and
Lang \cite{konczak2005voting}, as the candidates who win in every
(respectively, in at least one) completion of the given partial
preferences. Subsequent investigations produced a classification of
the computational complexity of the necessary and possible winners for
a variety of voting rules
\cite{DBLP:journals/ipl/BaumeisterR12,DBLP:journals/jcss/BetzlerD10,DBLP:journals/jair/XiaC11}.

Here, we bring forth a novel framework that aims to create bridges
between the COMSOC community and the data management community.  We
enrich the kinds of data analysis tasks that COMSOC methods currently
support by incorporating context about candidates, voters, issues, and
positions, thus going well beyond the mere determination of
winners. To achieve this, we accommodate COMSOC primitives within a
relational database framework, enabling the formulation and evaluation
of sophisticated queries.

\smallskip
\noindent{\bf Motivating Example.}
A \e{preference database} \cite{DBLP:journals/pvldb/JacobKS14} is
depicted in Figure~\ref{fig:example}. The relations $\relname{Cand}$
and $\relname{Voter}$ contain demographic information about political
candidates and voters, while $\relname{Supports}$ and
$\relname{Opposes}$ list positions of candidates on campaign issues,
and $\relname{Ballot}$ records results of an election or a poll.
Observe that $\relname{Ballot}$ specifies preferences of a voter in an
election with pairwise comparisons: the meaning of the tuple
$(\val{Oct-5},\val{Ann};\val{Clinton},\val{Trump})$ is that voter Ann
prefers Clinton to Trump when polled on October 5.

Preferences of voters may be incomplete. In particular, Ann states
that she prefers both Clinton and Johnson to Trump, but does not
specify a relative preference between Clinton and Johnson.  The
incomplete preference relation $\relname{Ballot}$ gives rise to four
completions, $B_1$ through $B_4$, in which each session is associated
with a complete ranking (a total order) over the candidates that is
consistent with the partial preference in $\relname{Ballot}$.

A data analyst may want to aggregate the votes of Ann and Bob to
determine the \e{winner} of the \val{Oct-5} election---the candidate
deemed most desirable by the voters---using a \e{voting rule}. In this
paper, we focus on \e{positional scoring rules}---voting rules that
assign a score to each candidate based on the candidate's position in
a ranking and then sum the scores across all rankings.

Our example involves two voting rules: plurality, which assigns a
score of $1$ to the top candidate in each ranking and $0$ to all other
candidates, and Borda, which assigns a score of $m-r$ to the candidate
at position $r$ out of $m$.  The sum of scores of each candidate in
each completion is shown in the bottom left table in
Figure~\ref{fig:example}.

We are concerned with answering the following kind of question: \e{Is
  there a winner according to the plurality rule in the October 5
  election who is pro-choice?} This question is phrased in logic-rule
style below as:
\begin{align*}
q_1() \dla & \Win(\val{plurality},\val{Oct-5},c), \\
           & \relname{Supports}(c,\val{pro-choice})
\end{align*}

What is the meaning of such a query posed on a database with partial
preferences? In Section~\ref{sec:framework}, we propose formal
semantics that generalizes the concepts of \e{necessary} and
\e{possible winners} from computational social choice to those of
\e{necessary} and \e{possible answers} in a partial preference
database.

Returning to the example in Figure~\ref{fig:example}, if we consider
Borda's rule, Clinton is the only necessary winner: she is the sole
winner in \relname{$B_1$}, \relname{$B_2$} and \relname{$B_3$}, and is
among the winners in \relname{$B_4$}. If we consider the plurality
rule, we find the following winners in each completion:
\relname{$B_1$}: Clinton, \relname{$B_2$}: Clinton and Trump,
\relname{$B_3$}: Clinton and Johnson, \relname{$B_4$}: Johnson and
Trump.  Consequently, the set of necessary winners under plurality is
empty.  Observe that, although no candidate is a necessary winner
under plurality, it is the case that \e{at least one} of the winners
in each completion is pro-choice.  Thus, under the plurality rule, the
query $q_1$ is \e{necessary}, i.e., \e{true} is a necessary answer of
$q_1$.

The preceding example illustrates the difference that context makes
and points to the richness brought by combining social choice and data
management.

\smallskip

\noindent {\bf Contributions.} At the conceptual level, we develop a
framework that combines social choice with database management, thus
making it possible to study social choice problems in the context of
additional information about voters, candidates, and issues. In
particular, we give rigorous semantics to queries in this framework by
introducing the notions of the \e{necessary answers} and the
\e{possible answers}.  At the technical level, we embark on an
investigation of the computational complexity of query evaluation in
this framework. In particular, we establish a number of results about
the necessary answers of queries that stand in sharp contrast to
results about the necessary winners under positional scoring rules.

We begin by exploring the complexity of computing the necessary
answers under the plurality rule. It is well known that there is a
polynomial-time algorithm for computing the necessary winners under
the plurality rule, and in fact, such an algorithm exists for every
pure positional scoring
rule~\cite{konczak2005voting,DBLP:journals/jair/XiaC11}.  We give a
polynomial-time algorithm for computing the necessary answers of
conjunctive queries that involve the plurality rule and have the
property that the winner atoms belong to different connected
components of the query. We prove that for a large class of
conjunctive queries, the property that winner atoms belong to
different connected components is \e{precisely} the property that
distinguishes tractable from intractable queries under the plurality
rule; particularly, evaluating queries in that class with two winner
atoms in the same connected component is coNP-complete.

Going beyond the plurality rule, we show that there is a natural
conjunctive query involving database relations and winners such that
computing the necessary answers is a coNP-complete problem for all
nontrivial positional scoring rules. By \e{trivial} we mean a scoring
rule that assigns the same score to every candidate (hence, every
candidate is a necessary winner). This result subsumes the
corresponding results (Theorems~4 and~5) that we have published in the
abridged version of this paper~\cite{ijcai18paper}.


\begin{figure*}[t!]
\small
\centering
\scalebox{0.9}{
\parbox{1.1\textwidth}{
\centering
\parbox[t]{1.1\textwidth}{
\begin{tabular}[t]{|c| c c c c|}
    \multicolumn{3}{l}{\relname{Cand} }\\
    \hline
    \w{\attname{cand}} & \w{\attname{party}} & \w{\attname{sex}} & \w{\attname{edu}} & \w{\attname{age}} \\
    \hline
    $\val{Clinton}$ & $\val{D}$ & $\val{F}$  & $\val{JD}$ & $\val{70}$\\
    $\val{Johnson}$ & $\val{L}$ & $\val{M}$  & $\val{BS}$ & $\val{64}$\\
    $\val{Trump}$ & $\val{R}$ & $\val{M}$  & $\val{BS}$ & $\val{71}$\\    
    \hline
  \end{tabular}
\quad
\centering
  \begin{tabular}[t]{|c c|}
\multicolumn{2}{l}{\relname{Supports} }\\
    \hline
    \w{\attname{cand}} & \w{\attname{issue}} \\
    \hline
    $\val{Clinton}$ & $\val{gun-control}$   \\
    $\val{Clinton}$ & $\val{pro-choice}$   \\
    $\val{Johnson}$ & $\val{pro-choice}$  \\
    \hline
  \end{tabular}
    \begin{tabular}[t]{|c c|}
\multicolumn{2}{l}{\relname{Opposes} }\\
    \hline
    \w{\attname{cand}} & \w{\attname{issue}} \\
    \hline
    $\val{Johnson}$ & $\val{gun-control}$  \\
    $\val{Trump}$ & $\val{gun-control}$  \\
    $\val{Trump}$ & $\val{pro-choice}$  \\
    \hline
  \end{tabular}
\quad
 \begin{tabular}[t]{|c| c c|}
    \multicolumn{3}{l}{\relname{Voter} }\\
    \hline
    \w{\attname{voter}} & \w{\attname{sex}} & \w{\attname{edu}} \\
    \hline
    $\val{Ann}$ & $\val{F}$ & $\val{MS}$  \\
    $\val{Bob}$ & $\val{M}$ & $\val{BS}$  \\
    \hline 
 \end{tabular}   
}
\vskip1em
{\parbox[t]{0.5\textwidth}{\centering
\begin{tabular}[t]{|c c|c c|}
 \multicolumn{4}{l}{\relname{Ballot}}\\
    \hline
    \w{\attname{election}} & \w{\attname{voter}} & \w{\attname{lcand}} & \w{\attname{rcand}}\\
    \hline
    $\val{Oct-5}$ & $\val{Ann}$ & $\val{Clinton}$ & $\val{Trump}$\\
    $\val{Oct-5}$ & $\val{Ann}$ & $\val{Johnson}$ & $\val{Trump}$\\
    $\val{Oct-5}$ & $\val{Bob}$ & $\val{Clinton}$ & $\val{Johnson}$\\
    $\val{Oct-5}$ & $\val{Bob}$ & $\val{Trump}$ & $\val{Johnson}$\\
    \hline
   \end{tabular}
 \vskip2em
\def\plurality{pl\xspace}
\def\Borda{Brd\xspace}
\def\twoapproval{2ap\xspace}
 \begin{tabular}[t]{|c | c c|c c|c c|c c|}
      \multicolumn{1}{c}{} & \multicolumn{2}{c}{\relname{$B_1$}} & \multicolumn{2}{c}{\relname{$B_2$}}
      & \multicolumn{2}{c}{\relname{$B_3$}} & \multicolumn{2}{c}{\relname{$B_4$}} \\
   \hline
   \w{\attname{cand}} & \w{\attname{\plurality}} & \w{\attname{\Borda}} &
     \w{\attname{\plurality}} & \w{\attname{\Borda}} & 
     \w{\attname{\plurality}} & \w{\attname{\Borda}} &
     \w{\attname{\plurality}} & \w{\attname{\Borda}} \\
   \hline
    $\val{Clinton}$ & $\val{2}$ & $\val{4}$ &  $\val{1}$ & $\val{3}$ &  $\val{1}$ & $\val{3}$ &  $\val{0}$ & $\val{2}$ \\
    $\val{Johnson}$ & $\val{0}$ & $\val{1}$ &  $\val{0}$ & $\val{1}$ &  $\val{1}$ & $\val{2}$ &  $\val{1}$ & $\val{2}$ \\
    $\val{Trump}$   & $\val{0}$ & $\val{1}$ &  $\val{1}$ & $\val{2}$ &  $\val{0}$ & $\val{1}$ &  $\val{1}$ & $\val{2}$ \\
   \hline
\end{tabular}
\vskip0.2em
Candidate scores in the possible completions under the plurality (\plurality) and Borda (\Borda) rules 
}}
{\parbox[t]{0.5\textwidth}{\centering
\begin{tabular}[t]{|c c|c|}
\multicolumn{3}{l}{A possible completion \relname{$B_1$} of \relname{Ballot}}\\
  \hline
     \w{\attname{election}} & \w{\attname{voter}} & \w{\attname{ranking}} \\
   \hline
    $\val{Oct-5}$ & $\val{Ann}$ & $\val{Clinton} \succ \val{Johnson} \succ \val{Trump}$ \\
    $\val{Oct-5}$ & $\val{Bob}$ & $\val{Clinton} \succ \val{Trump} \succ \val{Johnson}$ \\
   \hline
\end{tabular}
\vskip0.2em
\begin{tabular}[t]{|c c|c|}
\multicolumn{3}{l}{A possible completion \relname{$B_2$} of \relname{Ballot}}\\
  \hline
     \w{\attname{election}} & \w{\attname{voter}} & \w{\attname{ranking}} \\
   \hline
    $\val{Oct-5}$ & $\val{Ann}$ & $\val{Clinton} \succ \val{Johnson} \succ \val{Trump}$ \\
    $\val{Oct-5}$ & $\val{Bob}$ & $\val{Trump} \succ \val{Clinton} \succ \val{Johnson}$ \\
   \hline
\end{tabular}
\vskip0.2em
\begin{tabular}[t]{|c c|c|}
\multicolumn{3}{l}{A possible completion  \relname{$B_3$} of \relname{Ballot}}\\
  \hline
     \w{\attname{election}} & \w{\attname{voter}} & \w{\attname{ranking}} \\
   \hline
    $\val{Oct-5}$ & $\val{Ann}$ & $\val{Johnson} \succ \val{Clinton} \succ \val{Trump}$ \\
    $\val{Oct-5}$ & $\val{Bob}$ & $\val{Clinton} \succ \val{Trump} \succ \val{Johnson}$ \\
   \hline
\end{tabular}
\vskip0.2em
\begin{tabular}[t]{|c c|c|}
\multicolumn{3}{l}{A possible completion \relname{$B_4$} of \relname{Ballot}}\\
  \hline
     \w{\attname{election}} & \w{\attname{voter}} & \w{\attname{ranking}} \\
   \hline
    $\val{Oct-5}$ & $\val{Ann}$ & $\val{Johnson} \succ \val{Clinton} \succ \val{Trump}$ \\
    $\val{Oct-5}$ & $\val{Bob}$ & $\val{Trump\textbf{}} \succ \val{Clinton} \succ \val{Johnson}$ \\
   \hline
\end{tabular}
}}
}}
\caption{An example of a preference database.}
  \label{fig:example}
\end{figure*}

\section{Preliminaries}
\label{sec:prelim}

\paragraph{Relational databases and conjunctive queries.}
A \e{schema} is a collection of \e{relation symbols}, each having an
associated \e{signature}, which is a sequence of attribute names. A
\e{database} instantiates each relation symbol with a corresponding
relation (table).  We will use the database in
Figure~\ref{fig:example} as our running example.

A \e{query} is a function that maps every database into a
relation. More formally, a query has an associated input schema and an
output signature, and it maps every database over the input schema
into a relation over the output signature. If $D$ is a database and
$q$ is a query, then $q(D)$ denotes the relation resulting by
evaluating $q$ on $D$; each tuple in $q(D)$ is referred to as an
\e{answer} to $q$ on $D$.  In this paper, we study \e{conjunctive
  queries}, which correspond to the fragment of first-order logic
obtained from atomic formulas using conjunction and existential
quantification. Conjunctive queries are also known as
select-project-join (SPJ) queries, and are among the most frequently
asked database queries.

We will write queries as logic rules with a body and a head.  For
example, consider the following query:
 $$q(c) \dla  \relname{Voter}(\val{Ann},s,e), \relname{Cand}(c,p,s,a), a>\val{65}$$

 This query computes candidates who are older than 65 and whose sex is
 the same as that of voter \val{Ann}, and will return a single tuple,
 \val{Clinton}, when evaluated over the database in
 Figure~\ref{fig:example}.  Note that \val{Ann} and \val{65} are
 constants in $q$, while $a$, $c$, $e$, $p$, and $s$ are
 variables. The variables $e, p, s$ that occur in the body, but not in
 the head, of the query are existentially quantified.

 A \e{Boolean} query is a query that has no free variables, hence it
 stands for a yes/no (true/false) question about the database. For
 example, the Boolean query
$$q'() \dla  \relname{Voter}(\val{Ann},s,e), \relname{Cand}(c,p,s,a), a>\val{65}$$
asks whether or not there is a candidate who is older than 65 and
whose sex is the same as that of voter \val{Ann}. 

Conjunctive query evaluation has been a central topic of research in
the database management community (see
\cite{DBLP:books/aw/AbiteboulHV95}). In particular, it is well known
that, for every fixed conjunctive query $q$, there is a
polynomial-time algorithm that, given a database $D$, computes $q(D)$.

\paragraph{Incomplete databases and possible worlds.} Various notions
of database incompleteness have been studied in depth for several
decades. Common to these is the notion of \e{possible worlds}: these
include the completions of incomplete databases and the solutions in
data exchange and data
integration~\cite{FKMP05,DBLP:journals/jacm/ImielinskiL84,Lenzerini:2002:DIT:543613.543644}.
Query answering is a central challenge studied in these frameworks,
where the goal is to find the \e{certain} answers, i.e., the answers
obtained on every completion or on every solution.  Additionally, a
\e{possible answer} is an answer that is obtained on at least one
possible world.  More formally, if $\mathbf{W}$ denotes the set of
possible worlds of the database representation at hand, then the set
of certain answers to the query $q$ is the intersection
$\bigcap_{D\in\mathbf{W}}q(D)$, while the set of possible answers to
$q$ is the union $\bigcup_{D\in\mathbf{W}}q(D)$.

\paragraph{Voting profiles and voting rules.} Let
$C=\{c_1,\ldots,c_m\}$ be a set of \emph{candidates} (or
\emph{alternatives}) and let $V=\{v_1,\ldots,v_n\}$ be a set of
voters. A \emph{complete voting profile} is a tuple ${\bf
  T}=(T_1,\ldots,T_n)$, where each $T_i$ is a total order of the set
$C$ of candidates representing the ranking (preference) of voter $v_j$
on the candidates in $C$.

Positional scoring rules constitute a large and extensively studied
class of voting rules.  Each positional scoring rule on a set of $m$
candidates is specified by a scoring vector ${\bf a}=(a_1,\ldots,a_m)$
of non-negative integers, called the \emph{score values}, such that
$a_1\geq a_2\geq \ldots \geq a_m$. To avoid trivialities, we assume
that there are at least two different score values. Suppose that ${\bf
  T}=(T_1,\ldots,T_n)$ is a total voting profile. The score $s(T_i,c)$
of a candidate $c$ on $T_i$ is the value $a_s$ where $s$ is the
position of candidate $c$ in $T_i$.  When the positional scoring rule
$r$ is applied to ${\bf T}=(T_1,\ldots,T_n)$, it assigns to each
candidate $c$ the sum $\sum_{i=1}^ns(T_i,c)$ as the \e{score} of $c$.
The set $\mbox{W}(r, {\bf T})$ of the winners consists of the
candidates who achieved maximum score.

From now on, we focus on positional scoring rules that are defined for
every number $m$ of candidates. Thus, a \emph{positional scoring rule}
is an infinite sequence ${\bf a}_1, {\bf a}_2, \ldots, {\bf a}_m,
\ldots$ of scoring vectors such that each ${\bf a}_m$ is a scoring
vector of length $m$. Alternatively, a positional scoring rule is a
function $r$ that takes as argument a pair $(m,s)$ of positive
integers with $s\leq m$ and returns as value a non-negative integer
$r(m,s)$ such that $r(m,1) \geq r(m,2) \ldots \geq r(m,m)$. We will
also assume that the function $r$ is computable in polynomial time.
This implies that the winners can be computed in polynomial time.  

As examples, the \emph{plurality} rule is given by the infinite
sequence of scoring vectors of the form $(1,0,\dots,0)$, the
\emph{$2$-approval} rule is given by the infinite sequence of scoring
vectors of the form $(1,1,0,\dots,0)$, and the \emph{Borda} rule is
given by the infinite sequence of scoring vectors of the form
$(m-1,m-2,\dots,0)$.

Much on the literature makes the assumption that the rules are also
\emph{pure}, which means that the scoring vector ${\bf a}_{m+1}$ of
length $(m+1)$ is obtained from the scoring vector ${\bf a}_m$ of
length $m$ by inserting a score in some position of the scoring vector
${\bf a}_ m$, provided that the decreasing order of score values is
maintained. The plurality, $2$-approval and Borda rules are all pure.

\paragraph{Partial orders.} 
A \e{preference} over a collection of items is a linear order that
ranks the items from the most to the least preferred. Often, our
knowledge about the preference is only partial.  Missing information
in preferences is commonly modeled using a \e{partial order}, that is,
a relation that is reflexive, transitive, and antisymmetric, but not
necessarily total.  A \emph{completion} of a partial order is a total
order that extends that partial order. A partial order may have
exponentially many completions.

\paragraph{Partial voting profiles, necessary and possible winners.} A
\emph{partial voting profile} is a tuple ${\bf P}= (P_1,\ldots,P_n)$,
where each $P_i$ is a partial order of the set $C$ of candidates
representing the partial ranking (partial preference) of voter $v_j$
on the candidates.  A \emph{completion} of a partial voting profile
${\bf P}= (P_1,\ldots,P_n)$ is a complete voting profile ${\bf T}=
(T_1,\ldots,T_n)$ such that each $T_i$ is a completion of the partial
order $P_i$.  The notions of \emph{necessary} and \emph{possible}
winners were introduced by Konczak and Lang \cite{konczak2005voting}.

Let $r$ be a voting rule and ${\bf P}$ a partial voting profile.  The
set $\NW(r,{\bf P})$ of the \emph{necessary winners} with respect to
$r$ and $\bf P$ is the intersection of the sets $\W(r,{\bf T})$, where
$\bf T$ varies over all completions of $\bf P$.  In other words, a
candidate $c$ is a \emph{necessary winner} with respect to $r$ and
$P$, if $c$ is a winner in $\W(r,{\bf T})$ for every completion $\bf
T$ of $\bf P$.

The set $\PW(r,{\bf P})$ of the \emph{possible winners} with respect
to $r$ and $\bf P$ is the union of the sets $\W(r,{\bf T})$, where
$\bf T$ varies over all completions of $\bf P$.  In other words, a
candidate $c$ is a \emph{possible winner} with respect to $r$ and $\bf
P$, if $c$ is a winner in $\W(r,{\bf T})$ for at least one completion
$\bf T$ of $\bf P$.

On the face of the definitions, computing necessary and possible
winners requires exponential time, since, in general, a partial order
may have exponentially many completions.  There is a substantial body
of research on the computational complexity of the necessary and the
possible winners for a variety of voting rules. The following
\emph{complete} classification of the complexity of the necessary and
the possible winners for \emph{all} pure positional scoring rules was
obtained through the work of Konczak and
Lang~\cite{konczak2005voting}, Xia and
Conitzer~\cite{DBLP:journals/jair/XiaC11}, Betzler and
Dorn~\cite{DBLP:journals/jcss/BetzlerD10}, and Baumeister and
Rothe~\cite{DBLP:journals/ipl/BaumeisterR12}.

\begin{theorem} \label{class-thm} {\rm[Classification Theorem]}
The following hold.
\begin{itemize}
\item
  For every pure positional scoring rule $r$, there is a
  polynomial-time algorithm for computing the set $\NW(r,{\bf P})$ of
  necessary winners, given a partial voting profile $\bf P$.
\item  
  If $r$ is the plurality rule or $r$ is the veto rule, then there is
  a polynomial-time algorithm for computing the set $\PW(r,{\bf P})$
  of possible winners, given a partial voting profile $\bf P$.  For
  all other pure positional scoring rules, the following problem is
  NP-complete: given a partial voting profile $\bf P$ and a candidate
  $c$, is $c$ a possible winner with respect to $r$ and $\bf P$?
\end{itemize}
\end{theorem}

\section{A Relational Framework for COMSOC}
\label{sec:framework}

\paragraph{Preference schemas and sessions.} 
We adopt and refine the formalism for preference databases proposed in
\cite{DBLP:journals/pvldb/JacobKS14} and explored further in
\cite{DBLP:conf/pods/KenigKPS17}. A \e{preference} schema consists of
\e{preference} relation symbols and \e{ordinary} relation symbols.
The attribute list of each preference symbol $P$ is of the form
$(\beta; A_l, A_r)$, where $\beta$ is a list of attributes called
\emph{session signature}, and $A_l,A_r$ are attributes with candidates
as values. The intent is that if $({\bf b};c,d)$ is a tuple in an
instance of the preference symbol $P$, then candidate $c$ is preferred
to candidate $d$ in the session determined by $\bf b$.  In what
follows, we will assume that each session signature $\beta$ consists
of the attributes \attname{election} and \attname{voter}. The
intuition is that values for \attname{election} stand for particular
events, such as an election (e.g., the election for city council
members) or a poll (e.g., a poll on presidential candidates taken on
October 5), while the values for \attname{voter} range over the
possible voters.

\paragraph{Incorporating voting rules and winners.} We augment the
preference schema at hand with a new ternary relation symbol
$\Win$. The signature of $\Win$ is the triple
$(\attname{rule},\attname{election}, \attname{candidate})$. The intent
is that $\attname{rule}$ has voting rules as values, while
\attname{election} and \attname{candidate} have elections (or polls)
and candidates as values, respectively.  The new relation symbol
$\Win$ can now be used in standard relational database queries (e.g.,
SQL queries). To illustrate this, we return to our running example
(the preference database in Fig.~\ref{fig:example}) and give examples
of Boolean queries over this database.

\begin{itemize}
\item Query $q_1$: Is there a winner according to the plurality rule
  in the October 5 election who is pro-choice?
\begin{align*}
q_1() \dla & \Win(\val{plurality},\val{Oct-5},c), \\
           & \relname{Supports}(c,\val{pro-choice})
\end{align*}

\item Query $q_2$: 
  Are there a winner according to the Borda rule and a winner
  according to the 2-approval rule such that the former supports gun
  control and the latter opposes it?
\begin{align*}
  q_2() \dla  & \Win(\val{Borda},\val{Oct-5},c),\\
              & \Win(\val{2-approval},\val{Oct-5},d), \\
              & \relname{Supports}(c,\val{gun-control}), \\
              & \relname{Opposes}(d,\val{gun-control})
\end{align*}
\end{itemize}

\paragraph{Necessary and possible answers.} What is the semantics of
queries, such as the three preceding ones, in a framework that allows
for partial preferences, voting rules, and winners?  We propose two
different semantics of queries over partial preference databases,
namely, the \emph{necessary answers} and the \emph{possible answers}.

Let $D$ be a partial preference database. Then $D$ gives rise to a
partial voting profile ${\bf P}(D)$ consisting of the partial orders
that correspond to the identifiers of the session signature. For
example, the pair $(\val{Oct-5},\val{Ann})$ gives rise to the partial
order that expresses the (partial) preferences of Ann in the October 5
election (or poll).

Let $\Dc$ be a total preference database. We say that $\Dc$ is a
\emph{completion of $D$} if $\Dc$ is obtained from $D$ by completing
all partial preferences into total ones.  
Thus, $\Dc$ is a completion of $D$ iff $\Dc$ and $D$ agree on
  the ordinary relation symbols  and, for each preference
  relation symbol, the total voting profile ${\bf T}(\Dc)$ arising from
  $\Dc$ is a completion of the partial voting profile ${\bf P}(D)$
  arising from $D$. 
  
  \begin{definition} \label{NA-defn}{\rm
Let $q$ be a database query over the preference schema  augmented with relation symbol $\Win$.
\begin{itemize}
\item 
The \emph{necessary answers} of $q$ on $D$, denoted 
  $\NA(q,D)$, is the intersection $\bigcap q(\Dc)$ of the answers 
$q(\Dc)$ of $q$ on $\Dc$, as $\Dc$ ranges over all completions of $D$. 
\item
The \emph{possible answers} of $q$ on $D$, denoted
$\PA(q,D)$, is the union $\bigcup q(\Dc)$ of the answers $q(\Dc)$ of
$q$ on $\Dc$, as $\Dc$ ranges over all completions of $D$.
\end{itemize}
}
\end{definition}

If $q$ contains an atom of the form $\Win(r,e,c)$, then this atom is
evaluated on $\Dc$ by applying the voting rule $r$ on the total voting
profile ${\bf T}(\Dc)$. Thus, $\Win(r,e,c)$ evaluates to \emph{true}
on $\Dc$ if and only if $c$ belongs to the set $\W(r,{\bf T}(\Dc))$ of
the winners according to rule $r$ and the total voting profile ${\bf
  T}(\Dc)$. In this evaluation, only the part of the voting profile
${\bf T}(\Dc)$ that is associated with the election $e$ is needed.

Clearly, the preceding notions of necessary and possible answers
coincide with those of certain and possible answers in the framework
of incomplete databases, where the set $\cal W$ of possible worlds is
the set of all completions $\Dc$ of a given partial preference
database $D$.

If $q$ is a Boolean query, we say that $q$ is \e{necessary} \e{on $D$}
if $q(\Dc)$ is true for every completion $\Dc$, and \e{possible} \e{on
  $D$} if $q(\Dc)$ is true for at least one completion $\Dc$.  We
denote by $\necess q$ the decision problem: given $D$, is $q$
necessary on $D$? Similarly, we denote by $\possib q$ the decision
problem: given $D$, is $q$ possible on $D$?

We conclude this section by pointing out that our framework is
different from other approaches that have explored the interaction
between databases and social choice. For example,
Konczak~\cite{DBLP:conf/wlp/Konczak06} investigated the computation of
necessary and possible winners via logic programming (however, that
work does not involve the concepts of necessary answers and possible
answers of queries considered here).  In a different direction,
Lukasiewicz et al.~\cite{DBLP:journals/toit/LukasiewiczMST14}
investigated top-$k$ queries in databases using rankings whose
computation involve the aggregation of partial preferences.

\section{Complexity of Necessary Answers}
\label{sec:complexity}
How difficult is it to compute the necessary answers and the possible
answers of a fixed conjunctive query $q$, given a partial preference
database $D$?  As regards upper bounds, we can consider every
candidate tuple of values from the domain of $D$, and test whether
this tuple is indeed a necessary or possible answer. If the voting
rules occurring in $q$ are such that their winners are computable in
polynomial time, then, for every completion $\Dc$ of $D$, we have that
$q(\Dc)$ can be evaluated in polynomial time in the size of
$\Dc$. Consequently, deciding whether a tuple is a necessary answer of
$q$ is in coNP, while deciding whether it is possible is in NP.

Here, we will investigate the computational complexity of
the necessary answers 
of conjunctive queries. 
We begin by considering several  motivating examples.

\begin{example} \label{exam-1} {\rm Assume that $q$ is an atomic query
    of the form $\Win(r,e,c)$.  Computing the necessary and the
    possible answers of $q$ is the same as computing the necessary and
    the possible winners according to rule $r$. Thus, if $r$ is a pure
    positional scoring rule, then this query is accounted for by the
    Classification Theorem~\ref{class-thm} discussed earlier.}
\qed\end{example}
 
 \begin{example} \label{exam-2}{\rm 
 Let $q$ be the query $q_1$ encountered earlier:
\begin{align*}
q_1() \dla & \Win(\val{plurality},\val{Oct-5},c), \\
           & \relname{Supports}(c,\val{pro-choice})
\end{align*}
A moment's reflection reveals that there is a difference between 
the problem of deciding whether $q$ is possible and the problem of
deciding whether $q$ is necessary.

Indeed, to 
determine whether $q_1$ is possible 
on a partial preference database $D$, it is enough to compute the set
$\PW(\val{plurality},{\bf P}(D))$ of the possible winners with respect
to the plurality rule and the partial voting profile ${\bf P}(D)$,
and then intersect this set with the set of the pro-choice
candidates. Since the possible winners with respect to the plurality
rule are computable in polynomial time, it follows that 
$\possib{q_1}$ is decidable in polynomial time.

In other words, the
 possibility of $q_1$ 
can be rewritten to a
query that involves the possible winners $\PW(\val{plurality},{\bf
  T}(D))$. Concretely, 
$\possib{q_1}$ is equivalent to the query
\begin{align*}
q_1'() \dla & (c~\mbox{\tt IN}~\PW(\val{plurality},{\bf T}(D))),\\
           & \relname{Supports}(c,\val{pro-choice})\,.
\end{align*}

In contrast, 
$\necess{q_1}$ cannot
be rewritten (at least in a straightforward way) to a computation
involving the necessary winners $\NW(\val{plurality},{\bf T}(D))$. In
particular, 
$\necess{q_1}$
is \emph{not}
equivalent to the query
\begin{align*}
q_1''() \dla & (c~\mbox{\tt IN}~\NW(\val{plurality},{\bf T}(D))),\\
           & \relname{Supports}(c,\val{pro-choice})\,.
\end{align*}

Indeed,   for every completion $\Dc$
of $D$, there may exist a pro-choice winner (thus, $q_1$ is necessary on $D$),
but there may be
different winners  for different completions and, as a result,
there may exist no necessary winner who is also pro-choice (thus, the
query $q_1''$ evaluates to \emph{false} on $D$). 

Our results, however, will imply that $\necess{q_1}$ is decidable in
PTIME.  }
\qed\end{example}

\begin{example} \label{exam-4}{\rm 
Assume that $q$ is the query $q_3$: Are there two winners with different positions on at least one issue?
\begin{align*}
q_3() \dla &\Win(\val{plurality},e,c),\relname{Supports}(c,i),\\
          &\Win(\val{plurality},e,d),           
           \relname{Opposes}(d,i)
\end{align*}
Note that the query $q_3$ involves two $\Win$ atoms and that 
the candidates $c$ and $d$ occurring in these atoms are linked via the
variable $i$ in the two ordinary atoms of $q_3$.     

Our results will imply that $\necess{q_3}$ is coNP-complete. Thus,
Examples \ref{exam-2} and \ref{exam-4} demonstrate that, when bringing
together preferences, voting rules, and relational data in a unifying
framework, we are in a new state of affairs in which methods and
results from computational social choice need not apply directly.  }
\qed\end{example}

\subsection{Complexity Results for the Plurality Rule}
We present several complexity results for Boolean conjunctive queries.
In what follows in this subsection, we assume that all queries
considered involve the plurality rule and some fixed election, which
appears as a constant value $\val{elec}$ in the queries. Hence, each
$\Win$ atom has at most one variable, which stands for a winning
candidate; we refer to such variable as a \e{winner} variable.

\subsubsection{Tractability Results}
We begin with a tractability result.
\begin{theorem} \label{single-atom-thm} If $q$ is a Boolean
  conjunctive query consisting of a single $\Win$ atom and ordinary
  atoms, then $\necess q$ is decidable in polynomial time.
\end{theorem}
\begin{omitframe}
  We reduce the problem to the following generalization of the
  necessary-winner problem:  Given a partial voting profile and a
  subset $X$ of candidates, determine whether $X$ overlaps with the set
  of winners in every completion. In turn, we reduce this problem to
  maximum matching in a bipartite graph.  $\qedsymbol$
\end{omitframe}

Next, we prove Theorem~\ref{single-atom-thm}. Let $q$ be as in the
theorem. For a voting rule $r$, the evaluation of $q$ reduces to the
following problem that we refer to as \e{necessary intersection}:
Given a set $C$ of candidates, a partial voting profile $\tup P$ and a
party $A\subseteq C$, is $A$ represented in the winners in every
completion of $\tup P$? In other words, is it true that $\W(r,\tup
T)\cap A\neq\emptyset$ for all completions $\tup T$ of $\tup P$?

In the remainder of this proof, we present an algorithm for necessary
intersection in the case where $r$ is the plurality rule.

We assume that $C=\set{c_1,\dots,c_m}$ and $\tup P=(P_1,\dots,P_n)$.
For every $P_i$, we denote by $\max(P_i)$ the set of candidates $c$
such that $P_i$ does not prefer any other candidate over $c$. For each
candidate $c_j$, let $m_j$ be the number of partial profiles $P_i$
such that $c_j\in\max(P_i)$. The following lemma is straightforward.
\begin{lemma}\label{lemma:max-score}
  For all $j$ it holds that $m_j$ is the maximal score that $c_j$ can
  obtain over all possible completions of $\tup P$.
\end{lemma}
We select $c_q\in C\setminus A$ as a candidate outside of $A$ such
that $m_q$ is maximal among all candidates in $C\setminus A$.  A key
lemma is the following.
\begin{lemma}\label{lemma:A-max-score}
  For all completions $\tup U$ of $\tup P$ there exists a completion
  $\tup U'$ of $\tup P$ such that $\score(\tup U',c_q)=m_q$ and
  $\score(\tup U',a)\leq\score(\tup U,a)$ for all $a\in A$.
\end{lemma}
\begin{proof}
  Let $\tup U$ be a completion of $\tup P$. Transform $\tup U$ into $\tup
  U'$ by promoting $c_q$ to the head of the list whenever
  $c_q\in\max(P_i)$. Then the score of $c_q$ in $\tup U'$ is
  $m_q$. Moreover, by transforming $\tup U$ to $\tup U'$ no candidate
  $a\in A$ got any new point of score.
\end{proof}

We then conclude the following lemma.
\begin{lemma}\label{lemma:less-than-mq}
The following are equivalent.
\begin{enumerate}
\item There is a completion of $\tup P$ such that no party member 
  is a winner.
\item There is a completion $\tup T$ of $\tup P$ such that
  $\score(\tup T,a)<m_q$ for all $a\in A$.
\end{enumerate}
\end{lemma}
\begin{proof}
We prove each direction separately.

$1\Rightarrow 2:$\quad Let $\tup T$ be a completion of $\tup P$ such
that no candidate in $A$ is a winner. Let $c_e$ be a winner of $\tup
T$. Then $\score(\tup T,c_e)>\score(\tup T,a)$ for all $a\in
A$. According to Lemma~\ref{lemma:max-score} it holds that
$m_e\geq\score(\tup T,c_e)$, and hence, $m_q>\score(\tup T,a)$ for all
$a\in A$. So, we select as $\tup T$ the completion $\tup U'$ of
Lemma~\ref{lemma:A-max-score} for $\tup U=\tup T$.

$2\Rightarrow 1:$\quad Suppose that $\tup T$ is a completion of $\tup P$
such that $\score(\tup T,a)<m_q$ for all $a\in A$. Let $\tup T'$ be
obtained from $\tup T$ by applying Lemma~\ref{lemma:A-max-score}, where
$\tup T$ plays the role of $\tup U$ and $\tup T'$ the role of $\tup
U'$. Then in $\tup T'$ every party member of $A$ has a score strictly
lower than $c_q$, and therefore, none of the members of $A$ are
winners.
\end{proof}

From Lemma~\ref{lemma:less-than-mq} we conclude that we need to decide
on the existence of a completion $\tup T$ of $\tup P$ such that each
party member $a\in A$ has a score lower than $m_q$. This is done by
translating the problem into that of finding a maximum matching in a
bipartite graph, as follows.

Call a partial profile $P_i$ a \e{necessary supporter} of $A$ if
$\max(P_i)\subseteq A$. We construct the bipartite graph $G(U,V,E)$ as
follows.
\begin{itemize}
\item $U$ contains the node $i$ for every necessary supporter $P_i$ of
  $A$.
\item $V$ contains the nodes $\angs{a,1},\dots,\angs{a,m_q-1}$ for all
  party members $a\in A$.
\item $E$ connects $i\in U$ with $\angs{a,j}\in V$ whenever 
$a\in\max(P_i)$.
\end{itemize}

\begin{lemma}\label{lemma:matching}
The following are equivalent.
\begin{enumerate}
\item There is a completion $\tup T$ of $\tup P$ such that
  $\score(\tup T,a)<m_q$ for all $a\in A$.
\item $G$ has a matching of size $|U|$.
\end{enumerate}
\end{lemma}
\begin{proof}We prove each direction separately.

  $1\Rightarrow 2:$\quad Let $\tup T$ be a completion of $\tup P$ such
  that $\score(\tup U,a)<m_q$ for all $a\in A$. We construct a
  matching in $G$ by selecting for each $i$ a unique $\angs{a,j}$
  where $a$ is the top of $T_i$. Observe that we have enough
  $\angs{a,j}$ since $a$ got fewer top votes than $m_q$, as we are
  using the plurality rule.
  
  $2\Rightarrow 1:$\quad Suppose that $G$ has a matching of size
  $|U|$. We construct the desired $\tup T$ as follows. For each
  necessary supporter $P_i$, if $i$ is matched with $\angs{a,j}$
  then we select as $T_i$ a completion in which $a$ is at the top.
  For every other $P_i$ we select as $T_i$ a completion in which a
  candidate in $C\setminus A$ is at the top. Hence, each $a\in A$ gets
  the a score smaller than $m_q$ since there are only $m_q-1$ pairs
  $\angs{a,j}$.
\end{proof}
To conclude, our algorithm computes $m_q$, constructs the graph $G$,
and answers ``yes'' if and only if $G$ has \e{no} matching of size
$|U|$. The correctness is due to the combination of
Lemmas~\ref{lemma:less-than-mq} and~\ref{lemma:matching}.  This
completes the proof of Theorem~\ref{single-atom-thm}.

As a direct corollary of Theorem~\ref{single-atom-thm}, for the query
$q_1$ in Example~\ref{exam-2}, we have that $\necess{q_1}$ is in
polynomial time.  Another corollary, discussed next, generalizes
Theorem~\ref{single-atom-thm} through the notion of the \e{Gaifman
  graph} of a query, a notion that plays an important role in finite
model theory (see~\cite{DBLP:books/sp/Libkin04}); in this graph, the
nodes are the variables of the query and the edges consist of pairs of
variables occurring in the same ordinary atom.

The corollary applies to the case where every two distinct winner
variables belong to different connected components of the Gaifman
graph. In this case, we say that the $\Win$ atoms are \e{pairwise
  disconnected}.

\begin{corollary}
  \label{pairwise-disconnected} If $q$ is a Boolean conjunctive query
  with pairwise-disconnected $\Win$ atoms, then $\necess q$ is
  decidable in polynomial time.
\end{corollary}
\begin{proof}
  When the winner variables of a Boolean conjunctive query $q$ are
  pairwise disconnected, the query ``factors out'' over the subqueries
  $q'$ that correspond to the connected components of the Gaifman
  graph. Consequently, $q$ is necessary on $D$ if and only if each
  $q'$ is necessary on $D$. Hence, it suffices to solve $\necess{q'}$
  for each such $q'$. Theorem~\ref{single-atom-thm} implies that
  $\necess{q'}$ can be solved in polynomial time for each such $q'$
  since, by assumption, each such $q'$ contains either:
\begin{itemize}
\item one $\Win$ atom, in which case it covered by
  Theorem~\ref{single-atom-thm};
\item several copies of a $\Win$ atom, in which we can ignore all of
  these atoms except for one copy; \e{or}
\item no $\Win$ atoms at all, in which case the evaluation is that of
  a conjunctive query over an ordinary database.
\end{itemize}
Hence, $\necess q$ is decidable in polynomial time.
\end{proof}

\subsubsection{Dichotomy for Two $\Win$ Atoms}
Next, we show that, for a natural class of conjunctive queries, the
necessity of queries exhibit a PTIME vs.\ coNP-complete dichotomy.

\def\specialclass{\mathcal{C}_{\mathsf{2W}}}

\begin{definition} \label{class-c-defn} {\rm Let $\specialclass$ be
    the class of all Boolean conjunctive queries $q$ with the
    following properties:

(i) There are two distinct $\Win$ atoms (and both involve the plurality rule and the same fixed election).

(ii) All other atoms of $q$ are ordinary atoms such that no ordinary
relation symbol occurs twice (i.e., the ordinary atoms form a
\e{self-join free} query).  }
\end{definition}

\begin{theorem} \label{dichotomy-thm}
Let $q$ be a query in the class $\specialclass$.
\begin{itemize}
\item If the $\Win$ atoms of $q$ are pairwise disconnected, then
  $\necess q$ is decidable in polynomial time.
\item Otherwise, $\necess q$ is coNP-complete.
\end{itemize}
\end{theorem}
\begin{omitframe}
  Tractability follows from Corollary~\ref{pairwise-disconnected}.  We
  reduce $q_h$ to each remaining query via the proof technique of a
  dichotomy by Kenig et
  al.~\cite{DBLP:conf/pods/KenigKPS17}. $\qedsymbol$
\end{omitframe}

In the remainder of this section, we prove
Theorem~\ref{dichotomy-thm}.
We begin by proving coNP-hardness of $\necess{q_h}$ for a specific
Boolean conjunctive query  $q_h$:
\begin{align}
\label{eq:qh}
q_h() \dla & \Win(\val{plurality},\val{elec},c), R(c,d),\\
\notag           & \Win(\val{plurality},\val{elec},d)
\end{align}

\begin{lemma}\label{lemma:qh-hard}
  $\necess{q_h}$ is coNP-complete.
\end{lemma}
\begin{proof}
  Membership in coNP is straightforward: as a witness of
  non-necessity, we can use an appropriate completion of the partial
  preferences.  To prove coNP-hardness, we show a reduction from the
  complement of the \e{maximum independent set} problem. We are given
  as input a graph $g$ and a number $k$, and the goal is to determine
  whether $g$ has an independent set of size $k$. We construct a
  database as follows. 
\begin{itemize}
\item There is a single election, and the candidates are al pairs of
  the form $\angs{u,i}$ where $u$ is a node of $g$ and $1\leq i\leq
  k$.
\item There are $k$ voters $v_1,\dots,v_k$. The voter $v_i$ prefers
  every candidate $\angs{u,i}$ to all candidates $\angs{u',i'}$ where
  $i'\neq i$. Voter $v_i$ has no preference between the
  $\angs{u,i}$.
\item $R$ consists of all pairs $(\angs{u,i},\angs{u',i'})$ such that
  one of the following holds:
\begin{itemize}
\item $u=u'$ and $i\neq i'$;
\item $u\neq u'$ but $u$ is a neighbor of $u'$.
\end{itemize}
\end{itemize}
This concludes the reduction, and we complete the proof by proving its
correctness.

We have the following. If $g$ has an independent set
$U=\set{u_1,\dots,u_k}$ of size $k$, then we construct a completion
$\Dc$ of $D$ such that $q_h(\Dc)$ is false. For that, the voter $v_i$
positions $\angs{u_i,i}$ first. Then, the winners are the
$\angs{u_i,i}$, and then $q_h(\Dc)$ is indeed false. Similarly, if
$q_h$ is false in a completion, then the set of winners in this
completion forms an independent set of size $k$.
\end{proof}

We can now prove Theorem~\ref{dichotomy-thm}.

\vskip1em\par\noindent\textit{Proof (of Theorem~\ref{dichotomy-thm}).}\,
Tractability follows from Corollary~\ref{pairwise-disconnected}. For
coNP-hardness, we use the coNP-hardness of $q_h$, due to
Lemma~\ref{lemma:qh-hard}, and reduce $q_h$ to each remaining query
via the proof technique of a dichotomy by Kenig et
al.~\cite{DBLP:conf/pods/KenigKPS17}. We do so as follows.

Let $q$ be a conjunctive query in $\specialclass$ such that the $\Win$
atoms of $q$ are connected.  We write $q$ as follows.
  \begin{align*}
q() \dla & \Win(\val{plurality},\val{elec},c)\,,\,\\ 
& \Win(\val{plurality},\val{elec},d)\,,\,\\
&  \varphi(c,d,x_1,\dots,x_k)
\end{align*}

Here, $\varphi(c,d,x_1,\dots,x_k)$ is a conjunction of distinct atomic
queries, all ordinary database relations. We construct a reduction
from $q_h$ to $q$. That is, given a database $D_h$ for $q_h$, we
construct a database $D$ for $q$, so that the two queries $q$ and
$q_h$ are either both necessary or none is.

  Let $D_h$ be given. To construct $D$, we use the exact same partial
  profile as $D_h$, and we construct the other relations of $\varphi$
  as follows. We begin with empty relations. For each fact $R(a,b)$ of
  $D_h$, we add to $D$ all ground facts of $\varphi(a,b,...,b)$.

  To prove correctness, we show that for every completion of the
  profile, the two queries behave the same. For that, we show two
  claims:
  \begin{enumerate}
  \item If $R(a,b)$ is in $D_h$, then $\varphi(a,b,...,b)$ is true in
    $D$.
  \item If $\varphi(a,b,b_1,...,b_k)$ is true in $D$ for some
    $b_1,\dots,b_k$, then $R(a,b)$ is in $D_h$.
\end{enumerate}
The first claim is true by construction, and it requires neither the
assumption that $\varphi$ connects the winners nor that $q$ has no
self joins.

For the second claim, assume that $\varphi(a,b,b_1,...,b_k)$ is true
in $D$ for some $b_1,\dots,b_k$. If $c$ and $d$ occur in the same
atom, then it follows immediately that $R(a,b)$ is in
$D_h$. Otherwise, the variable tuples of the atoms include three
types: ones that involve $c$ and $x_i$s, ones that involve only
$x_i$s, and ones that involve $x_i$s and $d$.  Then the tuples of the
second and the third type have the property that they give rise to
reflexive tuples in $D$, where values that correspond to different
variables are the same. By connectedness, we get the existence of an
atom of the first type, in which $c$ takes the value $a$ and $x_i$
takes value $b$. This can only happen if we had $R(c,d)$ in $D_h$.
This completes the proof of Theorem~\ref{dichotomy-thm}.  \qed

\vskip1em

As a direct application, $\necess{q_3}$ is coNP-complete, where $q_3$
is the query in Example~\ref{exam-4}. In contrast, $\necess{q_4}$ is
solvable in polynomial time, where $q_4$ is the following query.
\begin{align*}
q_4() \dla & \Win(\val{plurality},\val{elec},c),
\\
& 
\Win(\val{plurality},\val{elec},d), 
          \\
          & \relname{Supports}(c,\val{pro-choice}),
           \relname{Cand}(d,p,\val{BS},a).
\end{align*}

\subsection{Hardness Beyond the Plurality Rule}
\begin{omitframe}
Here, we show that for a large class of positional scoring rules,
there are conjunctive queries $q$ involving Winner atoms and ordinary
atoms such that $\necess q$ is coNP-complete.

\begin{definition} \label{event-const-defn} Let $r$ be a pure
  positional scoring rule. We say that $r$ is \emph{eventually
    constant} if there is a positive integer $k$ such that for every
  $m$ and for every $s$ with $k< s \leq m$, we have that $r(m,s)=0$.
\end{definition}

The plurality rule, the $k$-approval rule, for $k\geq 2$, and the rule
$r$ with scoring vectors of the form $(2,1,0,\dots, 0)$ (this rule
played an important role in the proof of the Classification Theorem
\ref{class-thm}
\cite{DBLP:journals/ipl/BaumeisterR12,DBLP:journals/jcss/BetzlerD10})
are eventually constant positional scoring rules, while the Borda rule
is not.

\begin{theorem} \label{event-const-thm} If $r$ is an eventually
  constant positional scoring rule, then there is a Boolean
  conjunctive query $q$ involving Winner atoms and ordinary atoms such
  that $\necess q$ is coNP-complete.
\end{theorem}
\begin{proofhint}
By reduction from the complement of {\sc Positive $a$-in-$b$ Sat}, for suitable values of $a$ and $b$. This problem asks:
 given a CNF-formula consisting
  entirely of positive clauses of length $b$, is there a truth
  assignment such that  exactly $a$
  variables each clause are true?  If $b\geq 3$ and $1\leq a <b$, then
  {\sc Positive $a$-in-$b$ Sat}
   is NP-complete by Schaefer's
 dichotomy~\cite{Schaefer:1978:CSP:800133.804350}. $\qedsymbol$
\end{proofhint}

\begin{definition} \label{tail-const-defn} Let $k$ be a positive
  integer. The \emph{$k$-veto} rule (also known as the
  \emph{$(m-k)$-approval} rule) is the positional scoring rule with
  scoring vectors of the form $(1,\dots, 1, 0,\ldots, 0)$ with $(m-k)$
  scoring values of $1$ at the beginning and $k$ scoring values of $0$
  at the end.
\end{definition}
Clearly, the $1$-veto rule is the well known veto rule.

\begin{theorem} \label{tail-const-thm} For every $k\geq 1$, there is a
  Boolean conjunctive query $q$ involving Winner atoms of the $k$-veto
  rule and ordinary atoms such that $\necess q$ is coNP-complete.
\end{theorem}
\begin{proofhint}
Similar to  the proof of Theorem~\ref{event-const-thm}. $\qedsymbol$
\end{proofhint}
\end{omitframe}

A positional scoring rule $r$ is said to be \e{strict} if always
allows to strictly prefer one candidate over another, or more
formally, it is the case that $r(m,1)>r(m,m)$ for all $m>1$.  In this
section, we show the existence of a Boolean conjunctive query $q$ such
that $\necess{q}$ is coNP-complete for \e{all} strict positional
scoring rules.\footnote{We note that this result subsumes the
  corresponding results, namely Theorems~4 and~5, that we have
  published in the abridged version of this
  paper~\cite{ijcai18paper}. Specifically, our results there are
  restricted to the classes of \e{eventually constant} and $k$-veto
  rules.}

\def\qthr{q_{\mathsf{3w}}^r}

If $r$ is a positional scoring rule, then we denote by $\qthr$ the
following query that uses $r$ as a constant.
\begin{align}
\label{eq:qthr}  \qthr() \dla & \Win(r,e,x_1),
  \Win(r,e,x_2), \\ \notag
  & \Win(r,e,x_3), R(x_1,x_2,x_3)
\end{align}

\begin{theorem}\label{thm:all-hard}
  $\necess{\qthr}$ is coNP-complete for \e{all} strict positional scoring rules $r$.
\end{theorem}
In the remainder of this section, we prove Theorem~\ref{thm:all-hard}.
We show a reduction from 3-DNF Tautology. An instance of this problem
is a 3-DNF formula $\varphi=\psi_1\lor\dots\lor \psi_k$ over a
collection $x_1,\dots,x_\ell$ of Boolean variables, where each
$\psi_i$ is a conjunction of three atomic formulas. The goal is to
determine whether $\varphi$ is true fo every truth assignment to
$x_1,\dots,x_\ell$. So, let $r$ be a strict scoring rule, and let such
$\varphi$ be given.  We construct an instance of $\qthr$ as follows.

  There are $2\ell+1$ candidates:
  \begin{itemize}
  \item $x_i$ for $i=1,\dots,\ell$;
  \item $\neg x_i$ for $i=1,\dots,\ell$;
  \item a special candidate $x_0$.
    \end{itemize}
    We denote by $C$ the set of all candidates; that is:
    \[C=\set{x_0,x_1,\neg x_1,x_2,\neg x_2,\dots,x_\ell,\neg
      x_\ell}\,.\] We denote by $m$ the number $|C|$ of candidates,
    that is, $m=2\ell+1$. It is crucial for the proof that $m$ is odd,
    and this is the reason for having the special candidate $x_0$.  We
    fix a natural number $d\in\set{1,\dots,m-1}$ with the property
    that
\[r(m,d)>r(m,d+1)\,.\]

We identify a linear order over $C$ with a vector $\tup
c=(c_1,\dots,c_m)$ that consists of all the candidates of $C$, where
$c_1$ is the most preferred and $c_m$ is the least preferred. For
$i=1,\dots,\ell$, let $\tup c^i$ be the vector (linear order) that is
constructed as follows. Start with a vector of $m$ empty
placeholders. Next, place $x_i$ at position $d$ and $\neg x_i$ at
position $d+1$. Next, place all remaining $m-2$ elements at the
remaining $m-2$ placeholders, in an arbitrary order. We denote $\tup
c^i$ as follows.
\[\tup c^i=(c_1,\dots,c_{d-1},x_i,\neg x_i,c_d,\dots,c_{m-2})\]

For $i=1,\dots,\ell$, we define $m-2$ voters as follows. 
\begin{itemize}
\item The voter $v_i$ has the preference $\tup c^i$, with the
  preference $x_i>\neg x_i$ removed. Hence, $x_i$ and $\neg x_i$ are
  incomparable, but inserting either the preference $x_i>\neg x_i$ or the preference
  $\neg x_i>x_i$ results in a linear order.
\item For an \e{odd} $j\in\set{1,\dots,m-3}$, the voter $v_i^j$ has
  the total preference order $\tup c^i$, with all candidates but $x_i$
  and $\neg x_i$ shifted $j$ positions to the left in a circular
  manner; that is:
\[(c_1',\dots,c_{d-1}',x_i,\neg x_i,c'_d,\dots,c'_{m-2})\]
where $(c'_1,\dots,c'_{m-2})=(c_{1+j},\dots,c_{m-2},c_1,\dots,c_{j})$.
\item For an \e{even} $j\in\set{1,\dots,m-3}$, the voter $v_i^j$ has
  the same preference as an odd $j$, except that $x_i$ and $\neg x_i$
  switch positions; that is:
\[(c_1',\dots,c_{d-1}',\neg x_i,x_i,c'_d,\dots,c'_{m-2})\]
where $(c'_1,\dots,c'_{m-2})=(c_{1+j},\dots,c_{m-2},c_1,\dots,c_{j})$.
\end{itemize}
Taken over all $i=1,\dots,\ell$, the above defines $\ell(m-2)$
voters.

Denote the sequence $(x_1,\neg x_1,x_2,\neg x_2,\dots,x_\ell,\neg
x_\ell)$ as $(c''_1,\dots,c''_{m-1})$.  We define additional $m-1$
voters $u_1,\dots,u_{m-1}$, where the preference for $u_j$ is the
linear order $(c''_1,\dots,c''_{m-1},x_0)$ with $c''_1,\dots,c''_{m-1}$
shifted $j$ positions to the left in a circular manner; that is, the
preference of $u_j$ is given by the following vector.
\[(c''_{j+1},\dots,c''_{m-1},c''_1,\dots,c''_{j},x_0)\] This completes
our definition of the voters. All in all, we have precisely
$\ell(m-2)+m-1$ voters; we denote this number by $n$

\definecolor{Gray}{gray}{0.85}
\newcolumntype{G}{>{\columncolor{Gray}}l}
\newcolumntype{F}{>{$>$\quad}l}
\def\OC#1{\multicolumn{1}{F}{#1}}
\def\xz{\mbox{\fbox{$x_0$}}}

\begin{table*}[t]
\renewcommand{\arraystretch}{1.4}
\centering
  \caption{Example of voters for the variables $x_1$, $x_2$ and $x_3$\label{table:example-all-score}
and $d=3$}
\begin{tabular}{c|cF>{$>$\quad}G>{$>$\quad}GFFF}\hline
Voter & \multicolumn{6}{c}{Partial order}\\\hline
$v_1^1$ & $x_2$ & $\neg x_2$ & $x_1$ & $\neg x_1$  & $x_3$ & $\neg x_3$ & $\xz$\\
$v_1^2$ & $\neg x_2$ & $x_3$ & $\neg x_1$ & $x_1$   & $\neg x_3$ & $\xz$  & $x_2$ \\
$v_1^3$  & $x_3$ & $\neg x_3$  & $x_1$ & $\neg x_1$   & $\xz$  & $x_2$ & $\neg x_2$ \\
$v_1^4$  & $\neg x_3$ & $\xz$ & $\neg x_1$ & $x_1$     & $x_2$ & $\neg x_2$ & $x_3$\\
$v_1$ & $\xz$ & $x_2$ & \multicolumn{2}{G}{$>$\quad$\{x_1,\neg x_1\}$} & $\neg x_2$ & $x_3$ & $\neg x_3$\\
\hline
$v_2^1$ & $x_1$ & $\neg x_1$ & $x_2$ & $\neg x_2$  & $x_3$ & $\neg x_3$ & $\xz$\\
$v_2^2$ & $\neg x_1$ & $x_3$ & $\neg x_2$ & $x_2$   & $\neg x_3$ & $\xz$  & $x_1$ \\
$v_2^3$  & $x_3$ & $\neg x_3$  & $x_2$ & $\neg x_2$   & $\xz$  & $x_1$ & $\neg x_1$ \\
$v_2^4$  & $\neg x_3$ & $\xz$ & $\neg x_2$ & $x_2$     & $x_1$ & $\neg x_1$ & $x_3$\\
$v_2$ & $\xz$ & $x_1$ & \multicolumn{2}{G}{$>$\quad$\{x_2,\neg x_2\}$} & $\neg x_1$ & $x_3$ & $\neg x_3$\\
\hline
$v_3^1$ & $x_1$ & $\neg x_1$ & $x_3$ & $\neg x_3$  & $x_2$ & $\neg x_2$ & $\xz$\\
$v_3^2$ & $\neg x_1$ & $x_2$ & $\neg x_3$ & $x_3$   & $\neg x_2$ & $\xz$  & $x_1$ \\
$v_3^3$  & $x_2$ & $\neg x_2$  & $x_3$ & $\neg x_3$   & $\xz$  & $x_1$ & $\neg x_1$ \\
$v_3^4$  & $\neg x_2$ & $\xz$ & $\neg x_3$ & $x_3$     & $x_1$ & $\neg x_1$ & $x_2$\\
$v_3$ & $\xz$ & $x_1$ & \multicolumn{2}{G}{$>$\quad$\{x_3,\neg x_3\}$} & $\neg x_1$ & $x_2$ & $\neg x_2$\\
\hline
$u_1$ & $x_1$ & $\neg x_1$ & \OC{$x_2$} & \OC{$\neg x_2$} & $x_3$ & $\neg x_3$ & $\xz$\\
$u_2$ & $\neg x_1$ & $x_2$ & \OC{$\neg x_2$} & \OC{$x_3$} & $\neg x_3$ & $x_1$ & $\xz$\\
$u_3$  & $x_2$ & $\neg x_2$ & \OC{$x_3$} & \OC{$\neg x_3$} & $x_1$ & $\neg x_1$ & $\xz$\\
$u_4$  & $\neg x_2$ & $x_3$ & \OC{$\neg x_3$} & \OC{$x_1$} & $\neg x_1$ & $x_2$ & $\xz$\\
$u_5$   & $x_3$ & $\neg x_3$ & \OC{$x_1$} & \OC{$\neg x_1$} & $x_2$ & $\neg x_2$& $\xz$\\
$u_6$   & $\neg x_3$ & $x_1$ & \OC{$\neg x_1$} & \OC{$x_2$} & $\neg x_2$ & $x_3$ & $\xz$\\
\hline
\end{tabular}
\end{table*}

For a completion $\tup T=(T_1,\ldots,T_n)$ of our partial profile, let
us say that $\tup T$ \e{selects} $x_i$ if the completion of $v_i$
prefers $x_i$ to $\neg x_i$, and $\neg x_i$ if the completion of $v_i$
prefers $\neg x_i$ to $x_i$. The key lemma of the proof is the
following.

\begin{lemma}\label{lemma:winners-assignment}
  For every completion $\tup T$, the winners are precisely the atomic
  formulas selected $\tup T$.
\end{lemma}
\begin{proof}
  Let $\tup T$ be a completion, let $p\in\set{1,\dots,\ell}$, and let
  $a\in\set{x_p,\neg x_p}$ be an atomic formula. Let us compute the
  score $s(\tup T,a)$ that $\tup T$ assignes to $a$.
\begin{itemize}
\item For $i\neq p$, the voters $v_i$ and $v_i^1,\dots,v_i^{m-3}$
  jointly contribute to the score of $a$ the following portion that we
  denote by $\rho$.
\begin{equation}\label{eq:v-i-contrib}
\rho=\left(\sum_{i=1}^{m} r(m,i)\right) -r(m,d)-r(m,d+1)
\end{equation}
This is true because in the orders defined by $v_i^1,\dots,v_i^{m-3}$,
the atom $a$ appears precisely once in each position, except for the
positions $d$ and $d+1$.
\item The voter $v_p$ contributes to $a$ the following portion that we
  determine by $\sigma_a$.
\begin{equation}\label{eq:v-p-contrib}
\sigma_a=
\begin{cases}
r(m,d) & \mbox{if $\tup T$ selects $a$;}\\
r(m,d+1) & \mbox{otherwise.}\\
\end{cases}
\end{equation}
In particular, $\sigma_a>\sigma_{a'}$ if $a$ is selected and $a'$ is
not.
\item The voters $v_p^1,\dots,v_p^{m-3}$ jointly contribute to 
  the score $a$ the following portion.
\begin{align}\notag
&\frac{m-3}{2}\cdot r(m,d)+\frac{m-3}{2}\cdot r(m,d+1)=\\
&(\ell-1)(r(m,d)+r(m,d+1))\label{eq:v-p-j}
\end{align}
Here, we are using the fact that $m$ is odd, which is again due to the addition of
$x_0$.
\item The voters $u_1,\dots,u_{m-1}$ jointly contribute to the score
  of $a$ the following portion.
\begin{equation}\label{eq:u-j}
\sum_{i=1}^{m-1}r(m,i)
\end{equation}
\end{itemize} 
All in all, the score $s(\tup T,a)$ that $a$ gains in $\tup T$ sums up
to 
\begin{align}
\notag
s(\tup T,a)=&(\ell-1)\rho+\sigma_a+(\ell-1)(r(m,d)+r(m,d+1))\\
\label{eq:a}
&+\sum_{i=1}^{m-1}r(m,i)\,. 
\end{align}
 In particular, we conclude that the
score is determined only by whether $a$ is selected or not, and every
selected candidate has a higher score than every non-selected
candidate. 

To complete the proof, we need to show $x_0$ does not have a score
higher than the selected $a$. We will show that score $s(\tup T,x_0)$
is lower than $s(\tup T,a)$ for \e{every} atomic formula $a$, selected
or not. First, let us compute $s(\tup T,x_0)$. For $i=1,\dots,\ell$,
the voters $v_i$ and $v_i^1,\dots,v_i^\ell$ jointly contribute to
$x_0$ the portion $\rho$, and so, over all $i$ we get
$\ell\cdot\rho$. The voters $u_1,\dots,u_\ell$ jointly contribute to
the score of $x_0$ the portion $(m-1)\cdot r(m,m)$. Hence, we have the
following.
\begin{equation}\label{eq:x0}
s(\tup T,x_0)=\ell\cdot\rho+(m-1)\cdot r(m,m)
\end{equation}
Hence, combining~\eqref{eq:a} and~\eqref{eq:x0} we get the following
for all atomic formulas $a$.

\begin{align*}
&s(\tup T,a)-s(\tup T,x_0)=-\rho+\sigma_a+(\ell-1)(r(m,d)
\\&
\quad +r(m,d+1))+\left(\sum_{i=1}^{m-1}r(m,i)\right)-(m-1)\cdot r(m,m)
\\&
= -\left(\sum_{i=1}^{m} r(m,i)\right) +r(m,d)+r(m,d+1)+\sigma_a
\\&
\quad+(\ell-1)(r(m,d)+r(m,d+1))
\\&
\quad\,+\left(\sum_{i=1}^{m-1}r(m,i)\right)-(m-1)\cdot r(m,m)
\\&=-r(m,m)+\sigma_a+\ell(r(m,d)+r(m,d+1))
\\&
\quad
-(m-1)\cdot r(m,m)
\\&
= \sigma_a+\ell(r(m,d)+r(m,d+1))-m\cdot r(m,m)
\\&
> r(m,d+1)+\ell(r(m,d+1)\\&
\quad +r(m,d+1))-m\cdot r(m,m)
\\&
=m\cdot r(m,d+1)-m\cdot r(m,m)\geq 0
\end{align*}
Hence, $s(\tup T,a)-s(\tup T,x_0)>0$, as claimed. This completes the proof of the lemma.
\end{proof}

We complete the proof by defining the ternary relation $R$.  Recall
that $\varphi=\psi_1\lor\dots\lor \psi_k$. For $j=1,\dots,k$, let
$\psi_j=a_j^1\lor a_j^2\lor a_j^3$; we insert into $R$ the triple
$(a_j^1,a_j^2,a_j^3)$. In particular, $R$ contains $k$ tuples.  Based
on Lemma~\ref{lemma:winners-assignment}, the query $\qthr$ is true in
a completion if and only if the winners (i.e., the selected atomic
formulas) represent a satisfying assignment. Hence, the query $\qthr$
is necessary if and only if $\varphi$ is a tautology. This completes
the proof of Theorem~\ref{thm:all-hard}.

\section{Concluding Remarks}
\label{sec:conc}

We presented a framework that enriches social choice with relational
database context.  This framework supports the formulation of queries
about winners in elections, alongside contextual information about
candidates, voters, and positions on issues. In the presence of
incomplete voter preferences, the semantics of queries are given via
the notions of necessary and possible answers, which extend the
notions of necessary and possible winners. Our technical results about
the necessary answers of conjunctive queries reveal that the context
makes a substantial difference, since the complexity of the necessary
answers of queries may be higher than the complexity of the necessary
winners.

It remains open to determine the complexity of the possible answers
for the plurality rule and the veto rule (for all other positional
scoring rules, even computing the possible winners is an intractable
problem).  It is interesting to go beyond conjunctive queries and,
among others, consider queries that support aggregate operators, such
as count and average.
 
 The alternative modeling of preferences is another direction for
 future research. Probabilistic votes adopt statistical models of
 preferences, such as Mallows~\cite{Mallows1957} and the Repeated
 Insertion Model (RIM)~\cite{Doignon2004}.  The analog of computing
 necessary/possible winners is to compute the probability that a given
 candidate
 wins~\cite{DBLP:conf/aaai/BachrachBF10,DBLP:conf/aldt/LuB11}.  In our
 framework, the analog is \e{probabilistic query answering}, where the
 goal is to compute the marginal probability of possible query
 answers~\cite{DBLP:conf/vldb/DalviS04,DBLP:series/synthesis/2011Suciu}.
 Conjunctive query evaluation over RIM databases has been studied
 in~\cite{DBLP:conf/pods/KenigKPS17}, but without the angle of
 computational social choice.

The modeling of voter preferences may also incorporate \e{constraints}
(or \e{dependencies}) that restrict the possible completions to those
satisfying some conditions that are known to hold. An example is that
voters vote according to party affiliation---all candidates of one
party are preferred to all candidates of another party (but we do not
know upfront which party comes first).

\balance

\bibliographystyle{plain}
\bibliography{dbcomsoc}

\end{document}